\documentclass[11pt,letterpaper,dvipsnames]{amsart}
\usepackage[foot]{amsaddr}
\usepackage[margin=1.25in]{geometry}
\usepackage[square,numbers,sort&compress]{natbib}

\usepackage{amsmath,amsfonts,amssymb}
\usepackage{xcolor}
\usepackage{algorithm}
\usepackage[noend]{algpseudocode}
\usepackage{booktabs}
\usepackage{multirow}
\usepackage[hidelinks]{hyperref}
\usepackage{cleveref}
\usepackage{float}

\newcommand*{\doi}[1]{\href{http://dx.doi.org/#1}{doi: #1}}

\algrenewcommand\algorithmicrequire{\textbf{Input:}}
\algrenewcommand\algorithmicensure{\textbf{Output:}}
\algnewcommand\algorithmicforeach{\textbf{for each}}
\algdef{S}[FOR]{ForEach}[1]{\algorithmicforeach\ #1\ \algorithmicdo}

\newcommand{\R}{\mathbb{R}}
\newcommand{\N}{\mathbb{N}}

\newcommand{\Exp}{\mathbb{E}}
\newcommand{\E}[1]{\Exp[ #1 ]}

\newcommand{\BiggE}[1]{\Exp\Biggl[ #1 \Biggr]}
\newcommand{\bigE}[1]{\Exp\bigl[ #1 \bigr]}

\newcommand{\Prob}{\mathbb{P}}
\renewcommand{\P}[1]{\Prob[ #1 ]}

\newcommand{\BiggP}[1]{\Prob\Biggl[ #1 \Biggr]}
\newcommand{\bigP}[1]{\Prob\bigl[ #1 \bigr]}

\renewcommand{\v}{\mathbf{v}}
\newcommand{\s}{\mathbf{s}}
\newcommand{\x}{\mathbf{x}}
\newcommand{\y}{\mathbf{y}}
\newcommand{\z}{\mathbf{z}}
\newcommand{\p}{\mathbf{p}}
\renewcommand{\r}{\mathbf{r}}

\newcommand{\OPT}{\x^*}

\newcommand{\nalgInfGAP}{\textsc{InfeasibleGAP}}
\newcommand{\nalgFeasGAP}{\textsc{FeasibleGAP}}
\newcommand{\nalgImitGAP}{\textsc{ImitativeGAP}}
\newcommand{\nalgRandGAP}{\textsc{RandomGAP}}
\newcommand{\nalgFractionalKP}{\textsc{FractionalKnapsack}}

\newtheorem{theorem}{Theorem}
\newtheorem{proposition}{Proposition}
\newtheorem{lemma}{Lemma}

\vfuzz \hfuzz
\begin{document}

\title{Generalized Assignment and Knapsack Problems in the Random-Order Model} 

\author{Max Klimm$^1$}
\address{$^1$Institute for Mathematics, Technische Universität Berlin, Germany}
\email{klimm@math.tu-berlin.de}

\author{Martin Knaack$^1$}
\email{knaack@math.tu-berlin.de}

\begin{abstract}
We study different online optimization problems in the random-order model.
There is a finite set of bins with known capacity and a finite set of items arriving in a random order.
Upon arrival of an item, its size and its value for each of the bins is revealed and it has to be decided immediately and irrevocably to which bin the item is assigned, or to not assign the item at all.
In this setting, an algorithm is $\alpha$-competitive if the total value of all items assigned to the bins is at least an $\alpha$-fraction of the total value of an optimal assignment that knows all items beforehand.
We give an algorithm that is $\alpha$-competitive with $\alpha = (1-\ln(2))/2 \approx 1/6.52$ improving upon the previous best algorithm with $\alpha \approx 1/6.99$ for the generalized assignment problem and the previous best algorithm with $\alpha \approx 1/6.65$ for the integral knapsack problem.
We then study the fractional knapsack problem where we have a single bin and it is also allowed to pack items fractionally. For that case, we obtain an algorithm that is $\alpha$-competitive with $\alpha = 1/e \approx 1/2.71$ improving on the previous best algorithm with $\alpha = 1/4.39$. We further show that this competitive ratio is the best-possible for deterministic algorithms in this model.
\end{abstract}

\keywords{Generalized Assignment Problem, Random-Order Model, Knapsack Problem, Online Optimization} 

\maketitle

\newpage

\section{Introduction}

In the secretary problem, there is a sequence of $n$ applicants with unknown values that arrive in a random order.
Upon arrival of an applicant, the decision maker observes the value of the applicant and has to decide immediately and irrevocably whether to hire the applicant, or not. The famous $1/e$-rule stipulates that in order to maximize the probability of hiring the applicant with the highest value, the optimal policy is to not hire any of the first $n/e$ applicants and to hire any later applicant whose value exceeds all values observed so far (\citet{Dynkin63}; \citet{Lindley61}).
In the worst case, it achieves a probability of hiring the best applicant of $1/e$.
Surprisingly, this bound of $1/e$ also translates to a setting where the values of the applicants are drawn i.i.d.\ from a distribution not known to the decision maker and the objective is to maximize the expected value of the hired applicant (\citet[Theorem~2]{CorreaDFS22}).

In this paper, we study a generalization of the secretary problem known as the \emph{generalized assignment problem} in the random-order model. In this problem, there is a sequence of $n$ items that arrive in a random order and a set of $m$ bins. Every bin~$i$ has a known capacity $C_i > 0$. Upon arrival of  item~$j$, the value $v_{i,j} \geq 0$ and size $s_{i,j} \geq 0$ of item~$j$ for bin~$i$ is revealed.
After having observed these values, the decision maker has to decide immediately and irrevocably to which bin this item~$j$ is assigned to, or to not assign the item to any bin at all.
Formally, we introduce a binary variable $x_{i,j} \in \{0,1\}$ for each bin~$i$ and item~$j$ which is equal to $1$ if item~$j$ is assigned to bin~$i$ and $0$ otherwise. Since every item can be assigned to at most one bin, we have the inequality $\sum_{i=1}^m x_{i,j} \leq 1$ for all items~$j$.
In addition, the assignment of the items to the bins has to obey the capacity constraint of each bin, that is, $\sum_{j=1}^n s_{i,j} \, x_{i,j} \leq C_i$ for all bins~$i$.
The goal of decision maker is to maximize the expected value of the items assigned to the bins, i.e.,  $\smash{\sum_{i=1}^m \sum_{j=1}^n v_{i,j} \, x_{i,j}}$. Let
\textsc{Opt}
denote the value of an optimal offline solution of the generalized assignment problem, and let $x_{i,j}$, $i \in [m]$, $j \in [n]$ denote the variables set by an online algorithm. Then, for $\alpha \in [0,1]$, we call an online algorithm \emph{$\alpha$-competitive} if, on all instances, we have $\sum_{i=1}^m \sum_{j=1}^n v_{i,j} \,x_{i,j} \geq (\alpha - o(1))\,\textsc{Opt}$. We then also call $\alpha$ the \emph{competitive ratio} or \emph{competitiveness} of the algorithm.

The generalized assignment problem in the random-order model has been studied before by Kesselheim et al.~\cite{KesselheimRTV18} who gave an algorithm that is $1/8.1$-competitive.
\citet{NaoriR19} and, independently, \citet{AlbersKL21} gave an algorithm with an improved  competitive ratio of $1/6.99$.

\subsection{Our Contribution and Techniques}

In this paper, we give an algorithm for the online generalized assignment problem in the random-order model with an improved competitiveness of $(1-\ln(2))/2 \approx 1/6.52$. Similar to previous algorithms \cite{AlbersKL21,KesselheimRTV18,NaoriR19}, our algorithm is based on a fractional relaxation of the generalized assignment problem.
After a sampling phase of $n/2$ items that will not be assigned, after each item arrived, our algorithm solves the fractional allocation of the generalized assignment problem and assigns the item randomly according to the fractional variables corresponding to them item as long as the capacity of this bin is not exceeded.
A main challenge in the analysis of this algorithm is to bound the probability that the assignment of an item to a bin is successful, i.e., there is still enough capacity in the bin to accommodate the item.
Previous algorithms handle this challenge by distinguishing between large and small assignments.
\citet{KesselheimRTV18} call an assignment of an item~$j$ to a bin~$i$ \emph{heavy} if it uses more than half of the capacity of the bin, i.e., $s_{i,j} > C_i/2$, and \emph{light} otherwise.
Their algorithm then either only considers heavy assignments or only light assignments. If only heavy assignments are considered, the instance reduces to an edge-weighted matching problem in the random-order model for which an $1/e$-competitive algorithm is known (\citet{KesselheimRTV13}).
When only light assignments are considered, it suffices to bound the probability that in the assignment based on the fractional relaxation, the selected bin~$i$ is at most filled up to $C_i/2$.
\citet{AlbersKL21} and \citet{NaoriR19} use the same distinction between heavy and light assignments, but use a sequential approach. In a first phase of the algorithm from round $1$ to round $\lfloor 0.5261n \rfloor$ no item is packed. In a second phase from round $\lfloor 0.5261n \rfloor+1$ to round $ \lfloor 0.6906n\rfloor$ only heavy assignments are considered, and in a third phase from round $\lfloor 0.6906n\rfloor+1$ on only light assignments are considered. The heavy assignments are again handled by a bipartite matching algorithm while the small items are assigned based on the fraction relaxation.

Clearly, the distinction between heavy and light assignments degrades the performance of these algorithms since the packing of otherwise lucrative items maybe prevented based on the somewhat arbitrary distinction between heavy and light assignments.
Our algorithm does not use any distinction between heavy and light assignments at all. This is problematic for the heavy assignments since it is not clear how to bound the probability that the bin is sufficiently empty to pack the item. 
We circumvent this issue by considering preliminary assignments where the capacity of a bin may be violated by a single item. This  makes it possible to bound the probability of a successful preliminary assignment by bounding the probability that in previous rounds the sum of the weights of the items assigned to bin~$i$ is at most $C_i$ which can be done with similar techniques as by \citet{KesselheimRTV18}.
We show that this infeasible assignment obtains a competitive ratio of $1-\ln(2) \approx 1/3.26$.
To turn this infeasible assignment to a feasible one, we have to sacrifice an additional factor of $1/2$. With probability $1/2$, we run the original algorithm but do \emph{not} add the last item violating the capacity constraint to each bin. With the remaining probability of $1/2$, we run this algorithm only virtually and assign \emph{only} the last item that would violate the capacity of the bin in the original algorithm.
The resulting algorithm always obtains a feasible assignment and achieves a competitive ratio of $(1-\ln(2))/2 \approx 6.52$. The technique of overpacking a bin and selecting with probability $1/2$ either all items but the last or the last is reminiscent of the two-bin algorithm used by \citet{HanKM15} for the online knapsack problem with adversarial order and unit density.

We further study the fractional knapsack problem in the random-order model. This problem corresponds to the generalized assignment problem with a single bin where items can also be packed fractionally, i.e., we have $x_j \in [0,1]$ for all items~$j$. For this problem, we give an algorithm with competitive ratio $1/e$ improving on the previous best competitive ratio of $1/4.39$ due to \citet{GilibertiK21}.
Our algorithm has a sample phase of $\lfloor n/e \rfloor$ items. After that phase, in each round $\ell$, the algorithm solves the fractional knapsack problem once with all items revealed so far and once with all items revealed so far except the current item.
It then packs the current item fractionally corresponding to the total volume of items revealed during the sample phase that where removed from the fractional knapsack solution by including the current item.
This procedure ensures that the fractional knapsack solution computed by the algorithm does not exceed the knapsack capacity.
This general idea has been used before by the virtual algorithm devised by \citet{BabaioffIK07} for the $k$-secretary problem, but the latter problem is conceptionally easier since all items have a unit size and items are always packed integrally. 
We further show that the competitive ratio of $1/e$ is best-possible for deterministic algorithms. We first note that this does not readily follow from the lower bound of $1/e$ for the standard secretary problem since items may also be packed fractionally which may help an online algorithm.
For the proof, we exploit a similarity between deterministic algorithms for the fractional knapsack problem in the random-order model and randomized stopping rules together with Correa et al.~\cite[Theorem~2]{CorreaDFS22}. For a summary of our results and a comparison with previous results, see \Cref{tab:results}.

\begin{table}
\centering
\begin{tabular}{llr}
\toprule
\multirow{2}{*}{\textbf{Model}} & \multicolumn{2}{c}{Competitive Ratios}\\
 & \multicolumn{1}{c}{Previous Results} & \multicolumn{1}{c}{Our Results} \\
\midrule
Generalized Assignment Problem & $1/6.99$\; \cite{AlbersKL21,NaoriR19} & $\frac{1-\ln(2)}{2} \approx 1/6.52$\; [Thm.~\ref{theo:gap}]\phantom{$^\star$} \\[4pt]

Fractional Knapsack & $1/4.39$\; \cite{GilibertiK21} & $1/e \approx 1/2.71$\; [Thm.~\ref{theo:fractional-knapsack}]$^\star$\\[4pt]
\bottomrule
& & \footnotesize{$^\star$ best-possible by [Thm.~\ref{theo:lower}]}\\
\end{tabular}
\caption{Results obtained in this paper.
\label{tab:results}}
\end{table}

\subsection{Related Work}
For a general introduction to offline knapsack and generalized assignment problems, see \citet{MartelloT90}.
\citet{ChekuriK05} showed that the generalized assignment problem is $\mathsf{APX}$-hard.
For a minimization version of the problem, \citet{ShmoysT93} gave a $1/2$-approximation. \citet{ChekuriK05} noted that this $1/2$-approximation translates also to the maximization version of the problem. 
The generalized assignment problem contains as a special case the knapsack problem when there is a single bin only. This problem is $\mathsf{NP}$-hard (\citet{Karp72}) but admits a fully polynomial-time approximation scheme (FPTAS) as shown by \citet{IbarraK75}.
It further contains as a special case the multiple knapsack problem when the value and size of each item does not depend on the bin. \citet{ChekuriK05} showed that this problem does not admit an FPTAS even for 2 bins, unless $\mathsf{P} = \mathsf{NP}$, and gave a polynomial-time approximation scheme (PTAS). The generalized assignment problem further contains as a special case the maximum bipartite matching problem when all bin capacities are $1$, and all item sizes are $1$ or $2$.
A further special case of the generalized assignment problem is the AdWords problem that corresponds to the case when $s_{i,j} = v_{i,j}$ for all bins~$i$ and items~$j$. Here, the bins correspond to advertisers with a given daily budget $C_i$ and the items correspond to queries to a search engine. The value $s_{i,j} = v_{i,j}$ is the revenue generated from assigning query~$j$ to advertiser~$i$.

For the online bipartite matching problem where items arrive in an adversarial order, \citet{KarpVV90} gave an algorithm with a competitive ratio of $1-1/e \approx 0.63$. Interestingly, this is also the competitive ratio of the natural greedy algorithm in the random-order model. In the random-order model, an algorithm with a competitive ratio of $\alpha \approx 0.69$ was devised by \citet{MahdianY11}. 
The edge-weighted variant of the maximal online bipartite matching problem was first studied by \citet{KorulaP09} who gave a $1/8$-competitive algorithm. This was improved by \citet{KesselheimRTV13} to a $1/e$-competitive algorithm. This is best-possible since this contains the secretary problem as a special case.
The node-weighted variant of the maximal online bipartite matching problem corresponds to the matroid secretary problem when the underlying matroid is a transversal matroid. For general matroids, \citet{BabaioffIK07,BabaioffIKK18} gave a $O(1/\log k)$-competitive algorithm where $k$ is the rank of the matroid. This bound has been improved to $O(1/\sqrt{\log k})$ by \citet{ChakrabortyL12} and to $O(1/\log \log k)$ by \citet{Lachish14} and \citet{FeldmanSZ18}.
The question whether there is a constant-competitive algorithm for general matroids remains open \cite{BahraniBSW21}. Algorithms with constant competitiveness are only known for special cases of matroids such as a $1/(2e)$-competitive algorithm by Korula and P{\'{a}}l~\cite{KorulaP09}, a $1/4$-competitive algorithm by Soto et al.~\cite{SotoTV21}, and a $1/3.71$-competitive algorithm by Bérczi et al.~\cite{BercziLSV24} for graphic matroids, as well as a $1/9.6$-competitive algorithm by Ma et al.~\cite{Ma0W16}, a $5.19$-competitive algorithm by Soto et al.~\cite{SotoTV21}, a $1/4.75$-competitive algorithm by Huang et al.~\cite{HuangPZ24}, and a $1/3.26$-competitive algorithm by Bérczi et al.~\cite{BercziLSV24} for laminar matroids.
\citet{MehtaSVV05} study the online AdWords problem in the random-order model. Under a large markets assumption that no item contributes a constant fraction to the capacity of an advertiser, they provide an algorithm with a competitiveness of $1-1/e \approx 0.63$.

The random-order knapsack problem was first studied by Babaioff et al.~\cite{BabaioffIKK07} who gave an $\alpha$-competitive algorithm with $\alpha = 1/(10e) \approx 1/27.18$.
For the same problem, Kesselheim et al.~\cite{KesselheimRTV18} gave an algorithm with a competitiveness of $\alpha \approx 1/8.06$. The competitive ratio was further improved by Albers et al.~\cite{AlbersKL21} to $\alpha \approx 1/6.65$.
Under a large market assumption that no item contributes a constant fraction to the optimum solution, \citet{Vaze17} gave a $\alpha$-competitive algorithm with $\alpha = 1/(2e) \approx 5.44$. \citet{AbelsLSS22} considered the special case in which the capacity is $2$ and all items have a size of either $1$ or $2$. They obtain an algorithm that matches the best-possible competitiveness of $1/e$.

The fractional knapsack problem in the random-order model was first studied by \citet{KarrenbauerK20}. They developed a framework that turns an $\alpha$-competitive algorithm for the integral knapsack problem into an $1/(\alpha^{-1}+e)$-competitive algorithm for the fractional knapsack problem. Together with the algorithm of \citeauthor{AlbersKL21}, this yields a competitiveness of $\alpha \approx 1/9.37$. \citet{GilibertiK21} gave an algorithm with an improved competitive ratio of $\alpha \approx 1/4.39$.

The random-order online generalized assignment problem was first studied by \citet{KesselheimRTV18} who gave an $\alpha$-competitive algorithm with $\alpha \approx 1/8.06$. \citet{AlbersKL21} gave an algorithm with an improved competitiveness of $\alpha \approx 1/6.99$. A variant of the generalized assignment problem where capacities can be violated has been studied by \citet{FeldmanKMMP09} under a large markets assumption.

The generalized assignment problem can be further generalized to online packing LPs where the main difference is that packing an item may consume capacity of multiple bins. The competitiveness of online algorithms depends on various parameters of the problem such as the maximum number of bins used by packing an item and the ratio of the capacity of a bin and the capacity consumption of an item \cite{AgrawalWY14,BuchbinderN09,KesselheimRTV18}.

\section{Preliminaries}
\label{sec:preliminaries}

We let $\N$ denote the strictly positive natural numbers and, for $n \in \N$, let $[n] = \{1,\dotsc,n\}$. The generalized assignment problem (GAP) is given by a set $[m]$ of bins, each with a capacity $C_i > 0$ for $i \in [m]$, and a set $[n]$ of items. Assigning an item $j \in [n]$ to a bin $i \in [m]$ raises the total size of the bin by $s_{i,j} > 0$ and generates a profit or value of $v_{i,j} > 0$. The goal is to find an assignment of items to bins that maximizes the total value and fulfils the capacity constraint for every bin. The problem can be stated as the following integer linear program:
\begin{align}
    \max \ & \sum_{i=1}^{m} \sum_{j=1}^{n} v_{i,j} \, x_{i,j}, \nonumber \\
    \text{s.t. } 	& \sum_{j=1}^{n} s_{i,j} \, x_{i,j} \leq C_i &&\text{for all } i \in [m], \tag{C1} \label{eq:GapConstraint1} \\
    & \sum_{i=1}^{m} x_{i,j} \leq 1 &&\text{for all } j \in [n], \tag{C2} \label{eq:GapConstraint2} \\
    & x_{i,j} \in \{0,1\} &&\text{for all } i \in [m], j \in [n] \tag{C3} \label{eq:GapConstraint3}.
\end{align}
We also make use of the linear programming relaxation where \eqref{eq:GapConstraint3} is replaced by 
\begin{equation}
    \phantom{\max} \ x_{i,j} \geq 0 \qquad \text{ for all } i \in [m], j \in [n]. \tag{C4} \label{eq:GapRelaxationConstraint}
\end{equation}
We denote an assignment by $\x \in \{0,1\}^{m \times n}$ with $\x = (x_{i,j})_{i \in [m], j \in [n]}$ and we let 
\begin{equation*}
    v(\x) = \sum_{i=1}^{m} \sum_{j=1}^{n} v_{i,j} \, x_{i,j}
\end{equation*}
denote the total value of an assignment. We further let $\x^* \in \arg\max\{ v(\x) \mid \x \text{ fulfils } \eqref{eq:GapConstraint1}-\eqref{eq:GapConstraint3} \}$ and $\tilde{\x} \in \arg\max\{ v(\x) \mid \x \text{ fulfils } \eqref{eq:GapConstraint1}, \eqref{eq:GapConstraint2}, \eqref{eq:GapRelaxationConstraint} \}$ denote an optimal binary assignment and an optimal fractional assignment, respectively. For any assignment we let $\x_j$ be the $j$th column of $\x$. In addition let $\v_j = (v_{1,j},\dotsc,v_{m,j})^\top$ and $\s_j = (s_{1,j},\dotsc,s_{m,j})^\top$ be the vectors with values and sizes of item $j \in [n]$. Without loss of generality, we assume that $s_{i,j} \leq C_i$ for each bin $i \in [m]$ and each item $j \in [n]$.

In the random-order model an algorithm knows the number of items $n$, but the values $\v_j$ and sizes $\s_j$ of each item $j \in [n]$ are previously unknown. They are revealed to an algorithm one by one in rounds $1,\dotsc,n$ and an algorithm has to decide immediately and irrevocably whether to assign the item to a bin or to skip that item. The order in which the items are revealed is given by a random permutation $\pi \colon [n] \rightarrow [n]$ that assigns to each round $\ell \in [n]$ the item $\pi(\ell)$ that is revealed. The random permutation is drawn uniformly at random  from the set of all possible permutations on $[n]$, which we denote by $\Pi$. We say that an algorithm in the random-order model is $\alpha$-competitive, if for the assignment $\x$ of the algorithm we have $\E{v(\x)} \geq (\alpha - o(1)) \, v(\x^*)$ for every instance of the problem. The expectation is taken over the random permutation and over the randomized decisions of the algorithm.

\subsection{Knapsack and Fractional Knapsack}

The knapsack problem arises as a special case of GAP for $m=1$. Given are a knapsack capacity $C > 0$ and a set of items $[n]$ where each item $j \in [n] $ has a value $v_j > 0$ and a size $s_j > 0$. The goal is to find a subset of items that maximizes the total value while the total size does not exceed the capacity. We use the same notation as for GAP and denote a subset of items by a binary assignment $\x \in \{0,1\}^n$ with $\x = (x_1,\dotsc,x_n)$ and an optimal knapsack solution by $$\x^* \in \arg\max\Biggl\{ v(\x) \;\Bigg\vert\; \x \in \{0,1\}^n, \; \sum_{j=1}^n s_j \, x_j \leq C\Biggr\}.$$
Further, we define the density of an item $j \in [n]$ by $d_j = v_j/s_j$ and we assume without loss of generality that $ 0 < s_j  \leq C$ for all items $j \in [n]$.
We further assume that any ordering of the items by value or density yields a unique ordering. This could be achieved by arbitrary small perturbations of the values and does not affect our results, but makes the presentation easier. 
For the fractional variant of the knapsack problem, we use $\x \in [0,1]^n$ to denote a fractional assignment. We further assume that the items are sorted by their density in non-increasing order, i.e., $d_1 > d_2 > \cdots > d_n$.
Additionally, let $\rho = \max\{q \in [n] \mid \sum_{p=1}^{q} s_p \leq C \}$ denote the largest index, such that the first $\rho$ items do not exceed the knapsack capacity. By $\tilde{\x}$ we denote the fractional greedy solution which is defined by
\begin{equation*}
    \tilde{x}_j = \begin{cases} 1 & \text{if } j \leq \rho,\\ \dfrac{C - \sum_{q=1}^{\rho} s_q}{s_j} & \text{if } j=\rho+1, \qquad \text{for every } j \in [n]. \\ 0 & \text{otherwise},\end{cases}
\end{equation*}
Note that the fractional greedy solution is uniquely defined under our assumptions and that it defines an optimal solution of the fractional knapsack problem \cite[\S~17.1]{KorteV18}. We state some properties of fractional greedy solution for later reference.

\begin{proposition}
    \label{pro:fractional-greedy-solution}
    For the fractional greedy solution $\tilde{\x}$, it holds that
    \begin{enumerate}
        \item $\tilde{\x} = \arg\max\{ v(\x) \mid \x \in [0,1]^n, \; \sum_{j=1}^n s_j \, x_j \leq C\} $, \label{item:fractional-greedy-solution-optimality}
        \item $\sum_{j=1}^n s_j \, \tilde{x}_j = \min\{C, \; \sum_{j=1}^n s_j\}$. \label{item:fractional-greedy-solution-total size}
    \end{enumerate}
\end{proposition}

\section{The Generalized Assignment Problem}
\label{sec:gap}

\begin{algorithm}[t]
    \caption{\nalgInfGAP}
    \label{Alg:InfeasibleGAP}
    \begin{algorithmic}[1]
    \Require random permutation $\pi$, number of items $n$, set of bins $[m]$ with capacities $C_i, i \in [m]$
    \Ensure assignment $\x$ satisfying GAP constraints \eqref{eq:GapConstraint2}--\eqref{eq:GapConstraint3}
        \State $\x \gets 0$
        \State $Q_0 \gets \emptyset$
        \For{rounds $ \ell = 1,\dotsc,n $}
            \State $Q_\ell \gets Q_{\ell-1} \cup \{\pi(\ell)\}$
            \If{$ \ell > t \gets \lfloor n/2 \rfloor $}
                \State $ \tilde{\x}(Q_\ell) \gets $ optimal fractional assignment of revealed items $Q_\ell$
                \State $i^{(\ell)} \gets $ select a bin $i \in [m]$ where $ \P{i^{(\ell)} = i} = \tilde{x}_{i,\pi(\ell)}(Q_\ell) $ ($i^{(\ell)} = 0$ if none)
                \State Define $\x^{(\ell)}$ with $x_{i,j}^{(\ell)} = \begin{cases}
                    1 \text{ if } i=i^{(\ell)}, j=\pi(\ell), \\
                    0 \text{ otherwise }
                \end{cases}
                \text{ for } i \in [m], j \in [n]$.
                \If{$ i^{(\ell)} = 0 \textbf{ or } \sum_{k=1}^{\ell-1} s_{i^{(\ell)},\pi(k)} \, x_{i^{(\ell)},\pi(k)} \leq C_{i^{(\ell)}} $} 
                    \State $ \x \gets \x + \x^{(\ell)}$
                \EndIf
            \EndIf
        \EndFor
        \State \Return $\x$
    \end{algorithmic}
\end{algorithm}	

In this section, we study the generalized assignment problem in the random-order model. We improve over the currently best-known competitive ratio of $1/6.99$ shown independently by \citet{AlbersKL21} and \citet{NaoriR19}. We state our main result of this section in the following theorem.

\begin{theorem}
    \label{theo:gap}
    There exists an $\alpha$-competitive randomized algorithm for the generalized assignment problem in the random-order model where $\alpha = \frac{1-\ln(2)}{2} \approx \frac{1}{6.52}$.
\end{theorem}

We obtain the result by analysing a variant of an algorithm proposed by Kesselheim et al.~\cite{KesselheimRTV18}. The original algorithm consists of a sampling phase and an assignment phase. The sampling phase lasts for $t \in [n]$ rounds and during these rounds, the items are not assigned to a bin. Afterwards, in the assignment phase, the algorithm computes in each round $\ell > t$ the optimal fractional assignment $\tilde{\x}(Q_\ell)$ of all revealed items $Q_\ell$. Then it uses the fractional assignment of the current item $\pi(\ell)$ as a probability distribution over bins to randomly determine a bin $i^{(\ell)}$. Finally, the item $\pi(\ell)$ is assigned to the bin $i^{(\ell)}$ if this does not violate the capacity constraint of the bin. 

To achieve a constant competitive ratio for this algorithm, the authors of \cite{KesselheimRTV18} restrict the options $(v_{i,j}, s_{i,j})$ in the GAP instance to those with $s_{i,j} \leq C_{i}/2$. This allows to bound the probability that item $\pi(\ell)$ can be assigned to the bin $i^{(\ell)}$ by the probability that the total size of items assigned to the bin in previous rounds is at most $ C_{i^{(\ell)}}/2 $, which particularly is independent of $s_{i^{(\ell)},\pi(\ell)}$. The instance with the remaining options can be seen as an edge-weighted matching problem in the random-order model and is covered with a suitable algorithm. In the end, each algorithm is executed with a certain probability.

To improve the competitive ratio, Albers et al.~\cite{AlbersKL21} and Naori and Raz~\cite{NaoriR19} combined two algorithms for the different options by running them during different rounds of a single algorithm.

Our version of the algorithm, presented in \Cref{Alg:InfeasibleGAP}, does not restrict the instance. Instead, we allow the algorithm to exceed the capacity of every bin by at most one item. More precisely, we permit the assignment of item $\pi(\ell)$ to the bin $i^{(\ell)}$, if the total size of items assigned to the bin in previous rounds is at most $ C_{i^{(\ell)}} $. It turns out, that this is sufficient to handle all options $(v_{i,j}, s_{i,j})$ together, but the assignment returned by \Cref{Alg:InfeasibleGAP} may be infeasible. Therefore, we split the algorithm at the end of this chapter into two separate algorithms. The first algorithm does the same as \Cref{Alg:InfeasibleGAP}, but maintains feasibility by assigning items only if they do not violate the capacity constraint of the bin. Then, the second algorithm imitates the first algorithm and assigns the first item to each bin that the first algorithm could not assign due to the capacity constraint.

For the analysis of \Cref{Alg:InfeasibleGAP} let $\x$ denote the final assignment of the algorithm, let $t \in [n]$, and consider a fixed round $\ell > t$. As in the algorithm, we let $i^{(\ell)}$ denote the bin that the algorithm selects randomly and we let $\x^{(\ell)}$ denote the tentative assignment. Let $V^{(\ell)}$ denote the value that is obtained in round $\ell > t$. If the assignment of the tentative allocation $\x^{(\ell)}$ is successful, then $V^{(\ell)}$ equals $v(\x^{(\ell)})$, and if the assignment is not successful, then $V^{(\ell)}$ equals $0$. To bound the expected value of $V^{(\ell)}$, we determine the random order up to round $\ell$ in three steps as done by Kesselheim et al.~\cite{KesselheimRTV18}: (i) select a random subset $Q_\ell$ of $[n]$ with $ |Q_\ell| = \ell$ which are the items revealed in the first $\ell$ rounds; (ii) the item that is revealed in round $\ell$ is drawn uniformly at random from the set $Q_\ell$; (iii) the order of the remaining $\ell - 1$ items is determined by drawing an item uniformly at random from the remaining items.

By steps (i) and (ii) together with the random decision of the algorithm in round $\ell$, we can bound the expected value of the tentative assignment in round $\ell$ in terms of an optimal offline assignment. The crux is that the probability that the assignment is successful in round $\ell$ is bounded independently to the outcomes of steps (i) and (ii) and the random decision in round $\ell$. This allows us to bound the expected value of $V^{(\ell)}$ by the bound we derive for the probability that the assignment is successful times the expected value we obtain for the tentative assignment. In the following lemma, we obtain the bound for the value of the tentative assignment. 

\begin{lemma}
    \label{lemma:gap-fractional-value}
    For a round $\ell \in \{t+1,\dotsc,n\}$, let $Q_{\ell}$ be a random subset of $[n]$ with $|Q_{\ell}| = \ell$ and let $\pi(\ell)$ be an item drawn uniformly at random from $Q_{\ell}$. For the tentative assignment $\x^{(\ell)}$ we have
    \begin{equation*}
        \bigE{v(\x^{(\ell)})} \geq \frac{1}{n} \, v(\x^\star).
    \end{equation*}
\end{lemma}

\begin{proof}
    We first derive a bound for the optimal fractional assignment $ \tilde{\x}(Q_\ell) $ in terms of the optimal assignment $\OPT$. Since $Q_\ell$ is drawn uniformly at random from $[n]$, each item $j \in [n]$ is contained in $Q_\ell$ with a probability of $\frac{\ell}{n}$. We get
    \begin{align*}
        \bigE{v(\tilde{\x}(Q_\ell))} = \BiggE{\sum_{j \in Q_\ell} \v_j^\top \tilde{\x}_j(Q_\ell)} &\geq \BiggE{\sum_{j \in Q_\ell} \v_j^\top \x^\star_j} \\
        &= \sum_{j \in [n]} \P{j \in Q_\ell} \, \v_j^\top \x^\star_j = \frac{\ell}{n} \, v(\x^\star).
    \end{align*}
    For a fixed set $Q_\ell$ we can bound the value of the tentative assignment in round $\ell$ based on the randomly selected bin $i^{(\ell)}$ and the randomly drawn item $\pi(\ell)$. We have 
    \begin{align*}
        \bigE{v(\x^{(\ell)})} &= \sum_{j \in Q_\ell} \bigP{\pi(\ell) = j} \, \bigE{v(\x^{(\ell)}) \mid \pi(\ell) = j} \\
        &= \frac{1}{\ell} \sum_{j \in Q_\ell} \bigE{v(\x^{(\ell)}) \mid \pi(\ell) = j}
    \end{align*}
    and with
    \begin{align*}
        \bigE{v(\x^{(\ell)}) \mid \pi(\ell) = j} &= \sum_{i \in [m]} \bigP{i^{(\ell)} = i \mid \pi(\ell) = j} \, \bigE{v_{i^{(\ell)},\pi(\ell)} \mid \pi(\ell) = j, i^{(\ell)} = i} \\
        &= \sum_{i \in [m]} v_{i,j} \, \tilde{x}_{i,j}(Q_\ell),
    \end{align*}
    we get $\E{v(\x^{(\ell)})} = \frac{1}{\ell} \, v(\tilde{\x}(Q_\ell))$. Combining the equations completes the proof.
\end{proof}

The tentative assignment can be realized if the accumulated size of the selected bin has not exceeded the capacity in previous rounds. The next lemma shows a bound on the probability that the first $\ell-1$ assignments exceed the capacity of a bin $i \in [m]$. As in Kesselheim et al.~\cite{KesselheimRTV18}, we prove the bound for the tentative assignments in previous rounds instead of the realized assignments. The bound only depends on the order of the first $\ell - 1$ items and holds for every possible subset of items $Q_{\ell-1}$. Note that we slightly overload the notation here and use $Q_{\ell-1}$ as the set of items revealed in the first $\ell-1$ rounds and for the event that this set of items is revealed in the first $\ell - 1$ rounds.
 
\begin{lemma}
    \label{lemma:gap-success-bound}
    Consider a round $\ell \in \{t+1,\dotsc,n\}$ and a bin $i \in [m]$. Let $Q_{\ell-1}$ be any subset of $[n]$ with $|Q_{\ell-1}| = \ell-1$. We have 
    \begin{equation*}
        \BiggP{\sum_{k=t+1}^{\ell-1} s_{i, \pi(k)} \, x^{(k)}_{i, \pi(k)} > C_i \Biggm| Q_{\ell-1}} \leq \sum_{k=t+1}^{\ell-1} \frac{1}{k}.
    \end{equation*}
\end{lemma}

\begin{proof}
    Applying Markov's inequality on the probability, we get
    \begin{equation*}
        \BiggP{\sum_{k=t+1}^{\ell-1} s_{i, \pi(k)} \, x^{(k)}_{i, \pi(k)} > C_i \Biggm| Q_{\ell-1}} \leq \frac{1}{C_i} \, \BiggE{\sum_{k=t+1}^{\ell-1} s_{i, \pi(k)} \, x^{(k)}_{i, \pi(k)} \Biggm| Q_{\ell-1}}.
    \end{equation*}
    By linearity of expectation, we can consider each round $k \in \{t+1,\dotsc,\ell-1\} $ separately. Let $Q_k$ be an arbitrary subset of $Q_{\ell-1}$ with $|Q_k| = k$. Due to the random order, each item $j \in Q_k$ is revealed in round $k$ with probability $\frac{1}{k}$. We get
    \begin{align*}
        \bigE{s_{i, \pi(k)} \, x^{(k)}_{i, \pi(k)} \mid Q_k} &= \sum_{j \in Q_k} \bigP{\pi(k) = j \mid Q_k} \, \bigE{s_{i, \pi(k)} \, x^{(k)}_{i, \pi(k)} \mid Q_k, \pi(k) = j}\\
        &= \frac{1}{k} \sum_{j \in Q_k} s_{i,j} \, \bigE{x^{(k)}_{i, j} \mid Q_k, \pi(k) = j}.
    \end{align*}
    Since $\pi(k)=j$ the algorithm sets $x^{(k)}_{i, j} = 1$ if $i^{(k)} = i$ and $x^{(k)}_{i, j} = 0$ otherwise. Therefore, we get
    \begin{equation*}
        \bigE{x^{(k)}_{i, j} \mid Q_k, \pi(k) = j} = \bigP{i^{(k)} = i \mid Q_k, \pi(k) = j} = \tilde{x}_{i,j}(Q_k).
    \end{equation*}
    By the feasibility of the fractional assignment $\tilde{\x}(Q_k)$, we conclude that
    \begin{equation*}
        \bigE{s_{i, \pi(k)} \, x^{(k)}_{i, \pi(k)} \mid Q_k} = \frac{1}{k} \sum_{j \in Q_k} s_{i,j} \, \tilde{x}_{i,j}(Q_k) \leq \frac{C_i}{k},
    \end{equation*}
    proving the statement of the lemma.
\end{proof}

With the two lemmas at hand, we proceed to bound the expected value of the assignment returned by \Cref{Alg:InfeasibleGAP}.

\begin{lemma}
    \label{lemma:gap-infalg-bound}
    Let $\x$ be the infeasible assignment computed by \Cref{Alg:InfeasibleGAP}. Then, $\E{v(\x)} \geq (1 - \ln(2)) \, v(\OPT)$.
\end{lemma}

\begin{proof}
    Recall that $V^{(\ell)}$ denotes the expected value obtained in round $\ell > t$, i.e., $\E{v(\x)} = \sum_{\ell = t+1}^{n} \bigE{V^{(\ell)}} $. If the tentative assignment is realized in round $\ell$, then we have $V^{(\ell)} = v(\x^{(\ell)})$ and otherwise, we have $V^{(\ell)} = 0$. To handle the dependency between the event that the assignment is realized and the obtained value, we assume that the items in the first $\ell$ rounds are determined in three steps, as explained at the beginning of this section: In step (i), we determine the items that are revealed in the first $\ell$ rounds by drawing a random subset $Q_\ell$ with $ |Q_\ell| = \ell$ from $[n]$. In step (ii), we determine the item $\pi(\ell)$ that is revealed in round $\ell$ by drawing an item uniformly at random from $Q_\ell$. Finally, in step (iii), we determine the order in which the remaining $\ell-1$ items appear in the first $\ell-1$ rounds.
    
    By \Cref{lemma:gap-fractional-value}, we can bound the expected value of the tentative assignment in round $\ell$ only depending on step (i), step (ii) and the random decision of the algorithm in round $\ell$, i.e., the selected bin in round $\ell$. \Cref{lemma:gap-success-bound} shows for every bin $i$ that the assignments of the first $\ell - 1$ rounds exceed the capacity with a probability of at most $ \sum_{k=t+1}^{\ell-1} \frac{1}{k}$. This bound only uses the random order of the first $\ell - 1$ items. Therefore, it does not depend on steps (i) and (ii) and since it holds for each bin, it also does not depend on the selected bin in round $\ell$. Therefore, we get
    \begin{align*}
        \E{v(\x)} = \sum_{\ell = t+1}^{n} \bigE{V^{(\ell)}} \geq \sum_{\ell = t+1}^{n} \frac{1}{n} \, \biggl(1 - \sum_{k=t+1}^{\ell-1} \frac{1}{k}\biggr) \, v(\OPT).
    \end{align*}
    Let $H_k = \sum_{i=1}^k \frac{1}{i}$ denote the $k$th Harmonic number. Then, we obtain
    \begin{align*}
        \E{v(\x)} &\geq \sum_{\ell = t+1}^{n} \frac{1}{n} \, \biggl(1 - H_{\ell-1} + H_t\biggr) \, v(\OPT) \\
        &= \Biggl( \frac{n-t}{n}(H_t+1) - \frac{1}{n} \sum_{\ell = t}^{n-1}  H_{\ell} \Biggr) \, v(\OPT)\\
        &= \Biggl( \frac{n-t}{n}(H_t+1) - \frac{1}{n} \sum_{\ell = 1}^{n-1}  H_{\ell} + \frac{1}{n} \sum_{\ell = 1}^{t-1}  H_{\ell}\Biggr) \, v(\OPT).
    \end{align*}
    Next, we use the identity $\sum_{\ell=1}^{z} H_{\ell} = (z+1)H_z - z$ for all $z \in \mathbb{N}$ taken from \cite[p.~10]{GreeneK90} and obtain
    \begin{align*}
        \E{v(\x)} &\geq \biggl( \frac{n-t}{n}(H_t+1) -  H_{n-1} + \frac{n-1}{n} + \frac{t}{n}  H_{t-1} - \frac{t-1}{n}\biggr)\, v(\OPT) \\
        &= \biggl(2 - \frac{2t}{n} + H_t - H_n\biggr)\, v(\OPT).
    \end{align*}
    It remains to show that for all $n \in \mathbb{N}$, there is $t \in [n]$ such that $2 - \frac{2t}{n} + H_t - H_n \geq 1 - \ln(2)$.
    We claim that this inequality is satisfied for $t = \lfloor n/2 \rfloor$. Indeed, for even $n$, we obtain
    \begin{multline*}
        2 - \frac{2t}{n} + H_t - H_n = 1 + H_{n/2} - H_n = 1 - \sum_{\ell=n/2+1}^{n} \frac{1}{\ell} \\
        \geq 1 - \int_{n/2}^n \frac{1}{\ell} \,\text{d}\ell = 1 - \ln(n) + \ln(n/2) = 1 - \ln(2).
    \end{multline*}
    For odd $n$, we obtain
    \begin{equation*}
        2 - \frac{2t}{n} + H_t - H_n = 1 + \frac{1}{n} + H_{(n-1)/2} - H_n = 1 + H_{(n-1)/2} - H_{n-1}
    \end{equation*}
    and further
    \begin{multline*}
        1 + H_{(n-1)/2} - H_{n-1} = 1 - \sum_{\ell=\frac{n-1}{2}+1}^{n-1} \frac{1}{\ell} \\ \geq 1 - \int_{\frac{n-1}{2}}^{n-1} \frac{1}{\ell} \,\text{d}\ell = 1 - \ln(n-1) + \ln((n-1)/2) = 1 - \ln(2),
    \end{multline*}
    which shows the claimed result.
\end{proof}

Recall that \Cref{Alg:InfeasibleGAP} adds an item to a bin if the capacity constraint of the bin is not violated \emph{before} adding the item. In order to obtain a feasible assignment, we consider the variant where items are only added to a bin if the capacity constraint is not violated \emph{after} adding the item; see \Cref{Alg:FeasibleGAP} in \Cref{app:gap-algorithms} for a formal description. Like this, \Cref{Alg:FeasibleGAP} obtains a feasible assignment, but we potentially loose the values contributed by the last items added to each bin in \Cref{Alg:InfeasibleGAP}. To compensate for this loss, we run \Cref{Alg:FeasibleGAP} only with probability $1/2$ and with the remaining probability, we run \Cref{Alg:ImitativeGAP} in \Cref{app:gap-algorithms}. 
The algorithm tries to assign to each bin the last item that \Cref{Alg:InfeasibleGAP} assigned to the bin. In order to achieve this, it mimics \Cref{Alg:FeasibleGAP} by creating an assignment in the same way as \Cref{Alg:FeasibleGAP}. The actual assignment of \Cref{Alg:ImitativeGAP} is created by assigning only the first item to each bin that would violate the capacity constraint of the imitative assignment.
To obtain our result, we run both algorithms with a probability of $1/2$; see \Cref{Alg:RandomGAP}.

\begin{algorithm}[t]
    \caption{\nalgRandGAP}
    \label{Alg:RandomGAP}
    \begin{algorithmic}[1]
    \Require random permutation $\pi$, number of items $n$, set of bins $[m]$ with capacities $C_i, i \in [m]$
    \Ensure assignment $\hat{\x}$ satisfying GAP constraints \eqref{eq:GapConstraint1}--\eqref{eq:GapConstraint3}
        \State $ X \gets \text{Bernoulli}(1/2)$
        \If{$X = 1$}
            \State run \nalgFeasGAP
        \Else
            \State run \nalgImitGAP
        \EndIf
    \end{algorithmic}
\end{algorithm}

For the proof of \Cref{theo:gap} we show that both algorithms together are as good as \Cref{Alg:InfeasibleGAP}. The idea is the following: if we fix a permutation $\pi$ and if we fix the random decisions of the algorithms for that permutation, i.e., fixing the bin that is selected in each round, then the statement follows immediately. The rest follows from linearity of expectation. 

\begin{proof}[Proof of \Cref{theo:gap}]
    Let $\x,\y,\z,\hat{\x}$ be the output of \nalgInfGAP, \nalgFeasGAP, \nalgImitGAP, and \nalgRandGAP, respectively. We claim that the result holds for \nalgRandGAP. It runs \nalgFeasGAP\ and \nalgImitGAP\ each with a probability of $1/2$. By summing over all possible permutations, we get
    \begin{equation*}
        \E{v(\hat{\x})} = \frac{1}{2} \, \E{v(\y)+v(\z)} = \frac{1}{2} \sum_{\pi' \in \Pi} \P{\pi = \pi'} \, \E{v(\y)+v(\z) \mid \pi = \pi'}.
    \end{equation*}
    For a fixed permutation, we also know that the optimal fractional assignment $\tilde{\x}(Q_\ell)$ is fixed in each round $\ell > t$ of the algorithms \nalgFeasGAP\ and \nalgImitGAP. For the random decisions of the algorithms let $R$ be a vector of $n-t$ independent random variables $R_{t+1},\dotsc,R_{n}$ with $\P{R_\ell = i} = \tilde{x}_{i,\pi(\ell)}(Q_\ell)$ and $\P{R_\ell = 0} = 1 - \sum_{i=1}^{m} \tilde{x}_{i,\pi(\ell)}(Q_\ell)$ for each round $\ell > t$ and $i \in [m]$. Thus, bin $i$ is selected in round $\ell > t$, if $R_\ell = i$. Further, let $\mathcal{R}$ denote the set of all possible such vectors. We get
    \begin{multline*}
        \E{v(\y)+v(\z) \mid \pi = \pi'} =  \sum_{R' \in \mathcal{R}} \P{R = R' \mid \pi = \pi'} \, \E{v(\y)+v(\z) \mid \pi = \pi', R = R'}.
    \end{multline*}
    Consider the algorithms \nalgInfGAP, \nalgFeasGAP\ and \nalgImitGAP\ with a fixed permutation $\pi'$ and random decisions $R'$. For $i \in [m]$, let $\ell_i^*$ be the round where \nalgInfGAP\ assigns the current item into bin $i$, i.e., $R'_{\ell_i^*} = i$, but the assignment exceeds the capacity of bin $i$. For all previous rounds $\ell < \ell_i^*$ we have $x_{i,\pi'(\ell)} = y_{i,\pi'(\ell)}$, for round $\ell^*$ we have $ x_{i,\pi'(\ell^*)} = z_{i,\pi'(\ell^*)} = 1 $, and in all subsequent rounds $\ell > \ell^*$ we have $x_{i,\pi'(\ell)} = 0$. Should $\ell_i^*$ not exist, then we have $x_{i,\pi'(\ell)} = y_{i,\pi'(\ell)}$ for every round $\ell > t$. Since this holds for every bin, we have
    \begin{equation*}
        \E{v(\y)+v(\z) \mid \pi = \pi', R = R'} \geq \E{v(\x) \mid \pi = \pi', R = R'}.
    \end{equation*}
    We can undo the transformations applied to $\E{v(\y)+v(\z)}$, since the conditional probabilities for the selection of a bin are the same for each of the algorithms in each round. Therefore, we get
    \begin{equation*}
        \E{v(\y)+v(\z)} \geq \E{v(\x)}.
    \end{equation*}
    With \Cref{lemma:gap-infalg-bound} we conclude that \nalgRandGAP\ is $\frac{1-\ln(2)}{2}$-competitive, since
    \begin{equation*}
        \E{v(\hat{\x})} \geq \frac{1}{2} \, \E{v(\x)} \geq  \biggl(\frac{1-\ln(2)}{2}\biggr) \, v(\x^*),
    \end{equation*}
which shows the claimed result.
\end{proof}

\section{The Fractional Knapsack Problem}
\label{sec:fractional-knapsack}

In this section, we address the fractional variant of the knapsack problem in the random-order model. In this setting, the currently best-known competitive ratio is $\frac{1}{4.39}$ by Giliberti and Karrenbauer~\cite{GilibertiK21}. We state our main result of this section in the following theorem.

\begin{theorem}
    \label{theo:fractional-knapsack}
    There exists an $\alpha$-competitive algorithm for the fractional knapsack problem in the random-order model where $\alpha = \smash{\frac{1}{e}} \approx \smash{\frac{1}{2.71}}$.
\end{theorem}

\begin{algorithm}[t]
    \caption{\nalgFractionalKP}
    \label{Alg:FractionalKnapsack}
    \begin{algorithmic}[1]
        \Require random permutation $\pi$, number of items $n$, capacity $C$.
        \Ensure fractional knapsack solution $\x$
        \State $\x \gets 0$
        \State $ Q_0 \gets \emptyset $
        \For{rounds $ \ell = 1,\dotsc,n $}
        \State $ Q_\ell \gets Q_{\ell-1} \cup \{\pi(\ell)\}$
        \State $ \tilde{\x}(Q_\ell) \gets $ fractional greedy solution of revealed items $Q_\ell$
        \If{$ \ell > t \gets \lfloor n/e \rfloor $}
        \State $ x_{\pi(\ell)} = \tilde{x}_{\pi(\ell)}(Q_\ell) - \dfrac{\sum_{k=t+1}^{\ell-1}s_{\pi(k)} \, \bigl(\tilde{x}_{\pi(k)}(Q_{\ell-1}) - \tilde{x}_{\pi(k)}(Q_\ell) \bigr)}{s_{\pi(\ell)}} $
        \EndIf
        \EndFor
        \State \Return $\x$
    \end{algorithmic}
\end{algorithm}
 
We obtain the competitive ratio by \Cref{Alg:FractionalKnapsack}. The algorithm is built around the ideas of the virtual algorithm by Babaioff et al.~\cite{BabaioffIKK07}. The virtual algorithm was designed for packing items subject to a cardinality constraint of size $k$ in the random-order model, also referred to as the $k$-secretary problem. In each round $\ell$ after the sampling phase, the algorithm packs the current item if two conditions are fulfilled: the item has to be one of the $k$ most valuable items that have been revealed up to this round and the item that the current item removes from the set of the $k$ most valuable items must have been revealed during the sampling phase. Note that the second condition ensures feasibility, since each of the $k$ most valuable items after the sampling phase is removed at most once.

For the fractional knapsack problem we replace the set of the $k$ most valuable items by the fractional greedy solution. For each round $\ell$, let $\tilde{\x}(Q_\ell)$ denote the fractional greedy solution of the revealed items $Q_\ell$. For an arbitrary round $\ell$ and a fixed permutation $\pi$, we know that the total size of $\tilde{\x}(Q_\ell)$ cannot decrease compared to the total size of $\tilde{\x}(Q_{\ell-1})$ by part~(\ref{item:fractional-greedy-solution-total size}) of \Cref{pro:fractional-greedy-solution}, i.e., $ \sum_{k=1}^{\ell} s_{\pi(k)} \, \tilde{x}_{\pi(k)}(Q_\ell) \geq \sum_{k=1}^{\ell-1} s_{\pi(k)} \, \tilde{x}_{\pi(k)}(Q_{\ell-1})$ and we obtain
\begin{equation}
    \label{eq:fractional-knapsack-removed-size-bound}
    s_{\pi(\ell)} \, \tilde{x}_{\pi(\ell)}(Q_\ell) \geq \sum_{k=1}^{\ell-1} s_{\pi(k)} \, \bigl( \tilde{x}_{\pi(k)}(Q_{\ell-1}) - \tilde{x}_{\pi(k)}(Q_\ell) \bigr).
\end{equation}
Compared to the setting with a cardinality constraint, it is now possible that the revelation of item $\pi(\ell)$ in round $\ell$ removes multiple items from the fractional greedy solution or just a fraction of a single item. But note that it never increases the fraction by which other items are contained, i.e., $ \tilde{x}_{\pi(j)}(Q_{\ell-1}) \geq \tilde{x}_{\pi(j)}(Q_{\ell}) $ for each item $j \in Q_{\ell-1}$. 

Therefore, we have to adapt the packing conditions. In \Cref{Alg:FractionalKnapsack}, we pack the item $\pi(\ell)$ in round $\ell$ by a fraction $x_{\pi(\ell)}$, such that
\begin{equation*}
    s_{\pi(\ell)} \, x_{\pi(\ell)} = s_{\pi(\ell)} \, \tilde{x}_{\pi(\ell)}(Q_\ell) - \sum_{k=t+1}^{\ell-1} s_{\pi(k)} \, \bigl(\tilde{x}_{\pi(k)}(Q_{\ell-1}) - \tilde{x}_{\pi(k)}(Q_\ell) \bigr).
\end{equation*}
Thus, we take the size that the new item fills in the fractional greedy solution and for each item that was revealed after the sampling phase we subtract the size by which it is removed from the fractional greedy solution.

From the previous observations, it is easy to observe that $x_{\pi(\ell)} \in [0,1]$ and we get that the total size of the fractional assignment $\x$ is at most $C$, since
\begin{align*}		
    \sum_{\ell=1}^{n} s_{\pi(\ell)} \, x_{\pi(\ell)} &= \sum_{\ell=t+1}^{n} \Bigl( s_{\pi(\ell)} \, \tilde{x}_{\pi(\ell)}(Q_\ell) - \sum_{k=t+1}^{\ell-1}  s_{\pi(k)} \, \bigl(\tilde{x}_{\pi(k)}(Q_{\ell-1}) - \tilde{x}_{\pi(k)}(Q_\ell) \bigr) \Bigr) \\
    &= \sum_{\ell=t+1}^{n} \Bigl( s_{\pi(\ell)} \, \tilde{x}_{\pi(\ell)}(Q_\ell) - \sum_{k=\ell+1}^{n} s_{\pi(\ell)} \, \bigl(\tilde{x}_{\pi(\ell)}(Q_{k-1}) - \tilde{x}_{\pi(\ell)}(Q_k) \bigr) \Bigr) \\
    &= \sum_{\ell=t+1}^{n} s_{\pi(\ell)} \, \tilde{x}_{\pi(\ell)} \leq C,
\end{align*}
where we changed the summation order of the inner sum and the inequality at the end follows from the feasibility of fractional greedy solution $\tilde{\x}$ of all items $Q_n$. 

For the proof of \Cref{theo:fractional-knapsack}, we show in the following lemma a bound on the expectation of $x_j$ for each item $j \in [n]$.

\begin{lemma}
    \label{lemma:fractional-knapsack-item-bound}
    Let $\x$ be the fractional knapsack assignment of \Cref{Alg:FractionalKnapsack}. For each item $j \in [n]$ we have that
    \begin{equation*}
        \bigE{x_j} \geq \frac{\tilde{x}_j}{n} \sum_{\ell=t+1}^{n} \frac{t}{\ell-1}.
    \end{equation*}
\end{lemma}

\begin{proof}
    Since the permutation $ \pi $ is chosen uniformly at random, the probability that item $ j \in [n] $ is revealed in round $ \ell $ equals $ \frac{1}{n} $ for every round $ \ell $. We get
    \begin{equation*}
        \bigE{x_j} = \sum_{\ell=t+1}^{n} \bigP{\pi(\ell)=j} \, \bigE{x_{\pi(\ell)} \mid \pi(\ell)=j} = \frac{1}{n} \sum_{\ell=t+1}^{n} \bigE{x_{\pi(\ell)} \mid \pi(\ell)=j}.
    \end{equation*}
    For a fixed round $\ell > t$, we get by the assignment of \Cref{Alg:FractionalKnapsack} in round $\ell$ and by linearity of expectation that
    \begin{multline*}
        \bigE{x_{\pi(\ell)} \mid \pi(\ell)=j} = \bigE{\tilde{x}_{j}(Q_\ell) \mid \pi(\ell)=j} \\
         - \frac{1}{s_j} \, \BiggE{\sum_{k=t+1}^{\ell-1} s_{\pi(k)} \, \bigl(\tilde{x}_{\pi(k)}(Q_{\ell-1}) - \tilde{x}_{\pi(k)}(Q_\ell) \bigr) \;\Bigg\vert\; \pi(\ell)=j}. 
    \end{multline*}
    Due to the random order, each item contained in $Q_{\ell-1}$ is revealed in a round after the sample with probability $\tfrac{\ell-1-t}{\ell-1}$. Therefore, we get 
    \begin{multline*}
        \BiggE{\sum_{k=t+1}^{\ell-1} s_{\pi(k)} \, \bigl(\tilde{x}_{\pi(k)}(Q_{\ell-1}) - \tilde{x}_{\pi(k)}(Q_\ell) \bigr) \;\Bigg\vert\; \pi(\ell)=j} \\ 
        \begin{aligned}
            &= \frac{\ell-1-t}{\ell-1} \, \BiggE{\sum_{k \in Q_{\ell-1}} s_{k} \, \bigl(\tilde{x}_{k}(Q_{\ell-1}) - \tilde{x}_{k}(Q_\ell) \bigr) \;\Bigg\vert\; \pi(\ell)=j}\\
            &\leq \frac{\ell-1-t}{\ell-1} \, \bigE{s_{\pi(\ell)} \, \tilde{x}_{\pi(\ell)}(Q_\ell) \mid \pi(\ell)=j}   
        \end{aligned}
    \end{multline*}
    where the inequality follows from equation~\eqref{eq:fractional-knapsack-removed-size-bound}. We get
    \begin{align*}
        \bigE{x_{\pi(\ell)} \mid \pi(\ell)=j} &\geq \bigE{\tilde{x}_{j}(Q_\ell) \mid \pi(\ell)=j} - \frac{1}{s_j} \, \frac{\ell-1-t}{\ell-1} \, s_{j} \, \bigE{\tilde{x}_{j}(Q_\ell) \mid \pi(\ell)=j} \\
        &= \frac{t}{\ell - 1} \, \bigE{\tilde{x}_{j}(Q_\ell) \mid \pi(\ell)=j}
    \end{align*}
    and putting everything together yields
    \begin{align*}
        \bigE{x_{j}} \geq \frac{1}{n} \sum_{\ell=t+1}^{n} \frac{t}{\ell-1} \, \bigE{\tilde{x}_{j}(Q_\ell) \mid \pi(\ell)=j} \geq \frac{\tilde{x}_j}{n} \sum_{\ell=t+1}^{n} \frac{t}{\ell-1},
    \end{align*}
    as claimed.
\end{proof}

With the lemma, we are now able to prove \Cref{theo:fractional-knapsack}. The statement of the theorem follows from summing over all items and optimization over $t$.

\begin{proof}[Proof of \Cref{theo:fractional-knapsack}]
    We claim that the result holds for the fractional knapsack solution $\x$ of \Cref{Alg:FractionalKnapsack}. Applying \Cref{lemma:fractional-knapsack-item-bound} for every item $j \in [n]$ yields
    \begin{equation*}
        \E{v(\x)} = \sum_{j=1}^n v_j \, \E{x_j} \geq \sum_{j=1}^n v_j \, \frac{\tilde{x}_j}{n} \sum_{\ell=t+1}^{n} \frac{t}{\ell-1}  = \frac{t}{n} \, \sum_{\ell=t+1}^{n} \frac{1}{\ell-1} \, v(\tilde{\x}).
    \end{equation*}
    Bounding the sum implies
    \begin{equation*}
        \sum_{\ell=t+1}^{n} \frac{1}{\ell-1} = \sum_{\ell=t}^{n-1} \frac{1}{\ell} \geq \int_t^n \frac{1}{\ell} \ \mathrm{d}\ell = \ln\biggl(\frac{n}{t}\biggr).
    \end{equation*}
    With $ t = \bigl\lfloor \tfrac{n}{e} \bigr\rfloor $ we get 
    \begin{equation*}
        \E{v(\x)} \geq \biggl(\frac{1}{e} - o(1)\biggr) \, v(\tilde{\x})
    \end{equation*}
    and the result follows.
\end{proof}

We proceed to give an upper bound of $1/e$ on the competitiveness of any deterministic algorithm for the fractional knapsack problem in the random-order model. The proof exploits the similarity between deterministic algorithms for the fractional knapsack problem in the random-order model and randomized stopping rules together with Correa et al.~\cite[Theorem~2]{CorreaDFS22}. 

\begin{theorem}
    \label{theo:lower}
    Let $\varepsilon > 0$. No deterministic algorithm for the fractional knapsack problem in the random-order model is $(\frac{1}{e}+\varepsilon)$-competitive.
\end{theorem}

\begin{proof}
    For $n \in \mathbb{N}$, let $\mathcal{C}_n$ be the class of instances such that there are $n$ items, all items~$j$ have a size of $s_j = 1$ and the knapsack capacity is set to $C=1$. Moreover, for all items $j$, the value $v_j$ is determined by drawing a value $V_j$ from an unknown distribution, such that the values $v_1=V_1,\dotsc,v_n=V_n$ are independent and identically distributed (i.i.d.) random variables. Note that this is a special case of the random-order model.

    Those instances are closely related to the prophet problem with identical, but unknown distributions studied by Correa et al.~\cite{CorreaDFS22}. Here, an algorithm observes one by one $n$ nonnegative numbers $V_1,\dotsc,V_n$ which are drawn i.i.d.~from an unknown distribution. The algorithm has to decide when to stop such that it maximizes the expected value of the number at which it stops. 
    Hence, an algorithm for the prophet problem is a stopping rule $\r$ which decides for each $j \in [n]$ if it should stop at a current number $V_j$ based only on the values $V_1,\dotsc,V_j$ it has already observed. Formally, Correa et al.~\cite{CorreaDFS22} define a $n$-stopping rule $\r$ as a family of functions $r_1,\dotsc,r_n$ with $\smash{r_j \colon \R_+^j \rightarrow [0,1]}$ for all $j \in [n]$. For $\v = (v_1,\dotsc,v_n) \in \R_+^n$, $r_j(v_1,\dotsc,v_j)$ denotes the probability of stopping at $V_j$ under the conditions that we have not stopped at a number $V_1,\dotsc,V_{j-1}$ and that we observed the values $V_1=v_1,\dotsc,V_j=v_j$. Further, each $n$-stopping rule $\r$ defines a stopping time $\tau$ which is a random variable with support $\{1,\dotsc,n\}\cup \{\infty\}$ such that, for every $\v \in \R_+^n$,
    \begin{equation}
        \label{eq:frac-lower-time-def}
        \P{\tau = \ell \mid V_1=v_1,\dotsc,V_n=v_n} = \Biggl( \prod_{j \in [\ell-1]} (1-r_j(v_1,\dotsc,v_j)) \Biggr) \, r_\ell(v_1,\dotsc,v_\ell).
    \end{equation}
    The performance of a $n$-stopping rule is measured by the ratio~$\E{V_\tau} / \E{\max\{V_1,\dotsc,V_n\}}$, where we set $V_\infty=0$. We also set $\E{\v_\tau} = \E{V_\tau \mid V_1 = v_1,\dotsc,V_n=v_n} $ for every $\v \in \R_+^n $.    
    Correa et al.~\cite[Theorem~2]{CorreaDFS22} state the following: For every $\varepsilon > 0 $ there exists an $n_0 \in \N$, such that for every $n \geq n_0$ and every $n$-stopping rule $\r$ with stopping time $\tau$, there exists a unknown distribution $F$ such that
    \begin{equation*}
        \E{V_\tau} \leq \biggr(\frac{1}{e} + \varepsilon\biggl) \, \E{\max\{V_1,\dotsc,V_n\}},
    \end{equation*}
    where $V_1,\dotsc,V_n$ are i.i.d. random variables with distribution $F$.

    On the other hand, every deterministic algorithm for the fractional knapsack problem on instances in $\mathcal{C}_n$, defines an $n$-packing rule $\p$ which is a family of functions $p_1,\dotsc,p_n$ with $p_j \colon \R_+^j \rightarrow [0,1] $ for all $j \in [n]$, with the property that, for every $\v \in \R_+^n$, we have $ \sum_{j \in [n]} p_j(v_1,\dotsc,v_j) \leq 1$. Intuitively, $p_j(v_1,\dotsc,v_j)$ denotes the fraction that the algorithm packs of the item it observes in round $j$ under the condition that it observed the values $V_1=v_1,\dotsc,V_j=v_j$. The stated inequality holds since all sizes and capacity are equal to $1$ and since the algorithm returns a feasible fractional knapsack solution. 
    We let $ \E{\p(V)} $ denote the total expected value obtained by a deterministic algorithm for the fractional knapsack problem with $n$-packing rule $\p$ and, for $ \v \in \R_+^n $, we use $ \p(\v) = \sum_{j \in [n]} p_j(v_1,\dotsc,v_j) \, v_j$.

    We continue to show that every deterministic algorithm on an instance in $\mathcal{C}_n$ with $n$-packing rule $\p$ corresponds to a $n$-stopping rule $\r_\p$ with stopping time $\tau_\p$ for the prophet problem, such that $ \p(\v) = \E{v_{\tau_\p}}$ for every $\v \in \R_+^n $. 
    
    For each $\v \in \R_+^n $, we set $\P{\tau_\p = \ell \mid V_1=v_1,\dotsc,V_n=v_n} = p_\ell(v_1,\dotsc,v_\ell)$ for $\ell \in [n]$ and $ \P{\tau_\p = \infty \mid V_1=v_1,\dotsc,V_n=v_n} = 1 - \sum_{j \in [n]} p_j(v_1,\dotsc,v_j) $. Then, \eqref{eq:frac-lower-time-def} recursively defines the functions $ r_1,\dotsc,r_n$ of the corresponding stopping rule by
    \begin{align*}
        r_j(v_1,\dotsc,v_j) = 
        \begin{cases}
            \frac{p_j(v_1,\dotsc,v_j)}{1 - \sum_{k \in [j-1]} p_k(v_1,\dotsc,v_k)} &\text{if } \sum_{k \in [j-1]} p_k(v_1,\dotsc,v_k) < 1,\\
            0 &\text{if } \sum_{k \in [j-1]} p_k(v_1,\dotsc,v_k) = 1,
        \end{cases}
    \end{align*}
    for every $\v \in \R_+^n $. We obtain the property $ \p(\v) = \E{v_{\tau_\p}}$ for every $\v \in \R_+^n $ since
    \begin{align*}
        \E{v_{\tau_\p}} &= \sum_{\ell \in [n]} \P{\tau_\p = \ell \mid V_1=v_1,\dotsc,V_n=v_n} \, v_\ell \\
        &= \sum_{\ell \in [n]} p_\ell(v_1,\dotsc,v_\ell) \, v_\ell \\
        &= \p(\v).
    \end{align*}
    We conclude that for every $\varepsilon > 0$ there exists an $n_0 \in \N$, such that for every $n \geq n_0$ and every deterministic algorithm for the fractional knapsack problem on instances $\mathcal{C}_n$ with $n$-packing rule $\p$, there exists an unknown distribution $F$ such that
    \begin{equation*}
        \E{\p(V)} = \E{V_{\tau_\p}} \leq \biggr(\frac{1}{e} + \varepsilon\biggl) \, \E{\max\{V_1,\dotsc,V_n\}} = \biggr(\frac{1}{e} + \varepsilon\biggl) v(\tilde{\x}),
    \end{equation*}
    where we first used that $ \p(\v) = \E{v_{\tau_\p}} $ for every $\v \in \R_+^n $, for the inequality we used Correa et al.~\cite[Theorem~2]{CorreaDFS22} and at the end we used that $\E{\max\{V_1,\dotsc,V_n\}} = v(\tilde{\x})$ on an instance in $\mathcal{C}_n$.
\end{proof}

\newpage

\appendix

\section{Algorithms of \Cref{sec:gap}}
\label{app:gap-algorithms}

The following algorithms are referred to in \Cref{sec:gap}.

\begin{algorithm}[H]
    \caption{\nalgFeasGAP}
    \label{Alg:FeasibleGAP}
    \begin{algorithmic}[1]
    \Require random permutation $\pi$, number of items $n$, set of bins $[m]$ with capacities $C_i, i \in [m]$
    \Ensure assignment $\y$ satisfying GAP constraints \eqref{eq:GapConstraint1}--\eqref{eq:GapConstraint3}
        \State $\y \gets 0$
        \State $ Q_0 \gets \emptyset $
        \For{rounds $ \ell = 1,\dotsc,n $}
            \State $Q_\ell \gets Q_{\ell-1} \cup \{\pi(\ell)\}$
            \If{$ \ell > t \gets \lfloor n/2 \rfloor $}
                \State $ \tilde{\x}(Q_\ell) \gets $ optimal fractional assignment of revealed items $Q_\ell$
                \State $i^{(\ell)} \gets $ select a bin $i \in [m]$ where $ \P{i^{(\ell)} = i} = \tilde{x}_{i,\pi(\ell)}(Q_\ell)$ ($i^{(\ell)} = 0$ if none)
                \State Define $\x^{(\ell)}$ with $x_{i,j}^{(\ell)} = \begin{cases}
                    1 \text{ if } i=i^{(\ell)}, j=\pi(\ell), \\
                    0 \text{ otherwise }
                \end{cases}
                \text{ for } i \in [m], j \in [n]$.
                \If{$ i^{(\ell)} = 0 \textbf{ or } s_{i^{(\ell)},\pi(\ell)} + \sum_{k=1}^{\ell-1} s_{i^{(\ell)},\pi(k)} \, y_{i^{(\ell)},\pi(k)} \leq C_{i^{(\ell)}} $} 
                    \State $\y \gets \y + \x^{(\ell)}$
                \EndIf
            \EndIf
        \EndFor
        \State \Return $\y$
    \end{algorithmic}
\end{algorithm}

\begin{algorithm}[H]
    \caption{\nalgImitGAP}
    \label{Alg:ImitativeGAP}
    \begin{algorithmic}[1]
    \Require random permutation $\pi$, number of items $n$, set of bins $[m]$ with capacities $C_i, i \in [m]$
    \Ensure assignment $\z$ satisfying GAP constraints \eqref{eq:GapConstraint1}--\eqref{eq:GapConstraint3}
        \State $\z \gets 0$ \Comment{actual assignment}
        \State $\y \gets 0$ \Comment{imitative assignment}
        \State $ Q_0 \gets \emptyset $
        \For{rounds $ \ell = 1,\dotsc,n $}
            \State $Q_\ell \gets Q_{\ell-1} \cup \{\pi(\ell)\}$
            \If{$ \ell > t \gets \lfloor n/2 \rfloor $}
                \State $ \tilde{\x}(Q_\ell) \gets $ optimal fractional assignment of revealed items $Q_\ell$
                \State $i^{(\ell)} \gets $ select a bin $i \in [m]$ where $ \P{i^{(\ell)} = i} = \tilde{x}_{i,\pi(\ell)}(Q_\ell) $ ($i^{(\ell)} = 0$ if none)
                \State Define $\x^{(\ell)}$ with $x_{i,j}^{(\ell)} = \begin{cases}
                    1 \text{ if } i=i^{(\ell)}, j=\pi(\ell), \\
                    0 \text{ otherwise }
                \end{cases}
                \text{ for } i \in [m], j \in [n]$.
                \If{$ i^{(\ell)} = 0 \textbf{ or } s_{i^{(\ell)},\pi(\ell)} + \sum_{k=1}^{\ell-1} s_{i^{(\ell)},\pi(k)} \, y_{i^{(\ell)},\pi(k)} \leq C_{i^{(\ell)}} $} 
                    \State $ \y \gets \y + \x^{(\ell)}$
                \ElsIf{$ s_{i^{(\ell)},\pi(\ell)} + \sum_{k=1}^{\ell-1} s_{i^{(\ell)},\pi(k)} \, y_{i^{(\ell)},\pi(k)} > C_{i^{(\ell)}} $}
                    \If{$\sum_{k=1}^{\ell-1} z_{i^{(\ell)},\pi(k)}$ = 0}
                        \State $\z \gets \z + \x^{(\ell)}$
                    \EndIf
                \EndIf
            \EndIf
        \EndFor
        \State \Return $\z$
    \end{algorithmic}
\end{algorithm}

\newpage

\bibliographystyle{abbrvnat}
\bibliography{randomorderknapsack}

\end{document}